\documentclass{amsart}

\usepackage{amssymb,amsfonts,latexsym,amsmath,amsthm,wrapfig,graphicx, hyperref, 
}

\usepackage{bm}
\usepackage{epsfig}
\usepackage{mathtools}

\input xy
\xyoption{all}

\numberwithin{equation}{section}

\newcommand{\p}{\partial}

\newcommand{\be}{\begin{equation}}
\newcommand{\ee}{\end{equation}}

\newtheorem{theorem}{Theorem}[section]

\theoremstyle{definition}
\newtheorem{definition}[theorem]{Definition}

\newtheorem{remark}[theorem]{Remark}

\begin{document}

\title[Restricted isometric compression into low-dimensional varieties]{Restricted isometric compression of sparse datasets into low-dimensional varieties} 

\date{May 2024}

\author[V. Pop]{Vasile Pop}
\email{vasilepop@mail.usf.edu}
\address{4202 E. Fowler Ave., CMC342, Tampa, FL 33620}

\author[I. Teodorescu]{Iuliana Teodorescu}
\email{iuliana@usf.edu}
\address{4202 E. Fowler Ave., CMC342, Tampa, FL 33620}

\author[R. Teodorescu]{Razvan Teodorescu}
\email{razvan@usf.edu}
\address{4202 E. Fowler Ave., CMC342, Tampa, FL 33620}


\begin{abstract}
This article extends the known restricted isometric projection of sparse datasets in Euclidean spaces $\mathbb{R}^N$ down into low-dimensional subspaces $\mathbb{R}^k, k \ll N,$ to the case of low-dimensional varieties $\mathcal{M} \subset \mathbb{R}^N,$ of codimension $N - k = \omega(N)$. Applications to structured/hierarchical datasets are considered. 
\end{abstract}

\maketitle

\section{Introduction to Compressed Sensing}
In mathematical terms, a classical sparse recovery problem aims to recover a vector $\mathbf{x \in \Bbb R^N}$ from linear and underdetermined measurements $\mathbf{Ax=y}$ where $\mathbf{A \in \Bbb C^{m \times N}}$ models the linear measurement (information) process with $N$ being much larger than $m$. Looking closer, standard compressive sensing problem essentially identifies two questions which are not entirely independent. What matrices $\mathbf{A \in  \Bbb C^{m\times N}}$ are suitable and what are efficient reconstruction algorithms which recover the original signal $\mathbf{x}$ from measurements $\mathbf{y}$? 

 Even though there are several tractable strategies to solve standard compressive sensing problem, in this chapter we focus on the basis pursuit (also called $l_1$-minimization) strategy which consist in solving a convex optimization problem, more precisely, to find a minimizer for
$$\mathbf{\min\limits_{z \in \Bbb R^N} \|z\|_1 \text{ subject to } Az = y} \text{ where } \| \cdot\|_1 \text{ is } l_1 - \text{ norm}\ \ \ \ \ (P_1)$$ 

\subsection{Sparsity and Compressibility}
 Sparsity has been exploited in statistics and learning theory as a method for avoiding overfitting [15] and figures prominently in the theory of statistical estimation and model selection [16]. The notions of \textit{sparsity} and \textit{compressibility} are at the core of compressive sensing. 
 
\begin{definition}
A vector $\mathbf{x \in \Bbb C^N}$ is called $s$ - sparse if it has at most $s$ nonzero entries, i.e., if $$\mathbf{\|x\|_0} := card(\{j : x_j \neq 0, j = 1 \cdots , N\}) \leq s$$
\end{definition}
In applications, sparsity can be a strong constraint to impose, therefore we encounters vectors that are not exactly $s$ - sparse but compressible in the sense that they are well approximated by sparse ones. Informally, a vector $\mathbf{x \in \Bbb C^N}$ is \textit{compressible} if the error of its best $s$-term approximation $\sigma_s(\mathbf{x})_p$ decays quickly in $\mathbf{s}$.

\begin{definition} For $p>0$ the $l_p$-error of best $s$-term approximation to a vector $\mathbf{x} \in \Bbb C^N$ is defined by
$$\sigma_s(\mathbf{x})_p:=\inf\{\|\mathbf{x - z}\|_p, \mathbf{z} \in \Bbb C^N \text{ is a } s\text{-sparse}\}$$
\end{definition}

Note that sparsity is a highly nonlinear model, given a pair of $s$-sparse signals, a linear combination of two signals will in general no longer be $s$-spare. Sparsity is only a model and may not be the best fit for all applications, therefore extensions of sparsity are: block sparsity [20], join sparsity [19] and tree sparsity [21].

\subsection{Design Sensing matrices} 
Compressive sensing is not fitted for arbitrary matrices.  It is still open problem to construct explicit matrices which are provably optimal, but a breakthrough is achieved by resorting to random matrices. Examples of random matrices are \textit{Gaussian} matrices  whose columns consist of independent random variables following a standard normal distribution and \textit{Bernoulli} matrices whose columns are
independent random variables taking the values +1 and -1 with equal probability. We now list a number of desirable  conditions that a sensing matrix $\mathbf{A}$ should have to guaranty recovery of sparse vectors.

\begin{definition} A matrix $\mathbf{A \in \Bbb K^{m \times N}}$  (where $\Bbb K$ is $\Bbb R$ or $\Bbb C$) is said to satisfy the \textit{null space property} relative to a set $S \subset [N]$ if $$\mathbf{\|v_S\|_1 < \|v_{\overline{S}}}\|_1 \text{ for all }\mathbf{v} \in Ker \mathbf{A} - \{0\}$$ where $\mathbf{v}_S$ is the restriction of $\mathbf{v}$ on the indices in $S$. It is said to satisfy the \textit{null space property of order $s$} if it satisfies the null space property relative to any set $S \subset [N]$ with $card(S)\leq s$.
\end{definition}

A slightly strengthened versions of the null space property are needed to reconstruct scheme with respect to sparsity defect (stable null space property) or a scheme affected by error (robust null space property).

\begin{definition} A matrix $\mathbf{A} \in \Bbb C ^{m \times N}$ is said to satisfy the \textit{stable null space property} with constant $0 < \rho <1$ relative to a set $S \subset [N]$ if
$$\|\mathbf{v}_S\|_1 \leq \rho \|\mathbf{v}_{\overline{S}} \|_1 \text{ for all } \mathbf{v} \in Ker \mathbf{A}$$
It is said to satisfy the \textit{stable null property of order $s$ }with constant $0 < \rho < 1$ if it satisfy the stable null space property with constant $0 < \rho < 1$ relative to any set $S \subset [N]$ with $card(S) \leq s$.
\end{definition} 

\begin{definition} The matrix $\mathbf{A} \in \Bbb C^{m \times N}$ is said to satisfy the \textit{robust null space property} (with respect to $\| \cdot \|$) with constants $0 < \rho < 1$ and $\tau >0$ relative to a set $S \subset [N]$ if
$$\|\mathbf{v}_S\|_1 \leq \rho \|\mathbf{v}_{\overline{S}} \|_1 + \tau \|\mathbf{Av}\| \text{ for all } \mathbf{v} \in \Bbb C^N$$ It is said to satisfy the \textit{robust null space property of order $s$} with constants $0 < \rho < 1$
and $\tau > 0$ if it satisfies the robust null space property with constants $\rho, \tau$ relative to any set $S\subset [N]$ with $card(S)\leq s$.
\end{definition}

It has been proved in [7] (Chapter 4) that \textit{null space property}, \textit{stable null space property} and \textit{robust null space property} are necessary and sufficient condition for extract recovery of sparse vectors via basis pursuit program $(P_1)$.

Since the null space property is not easily verifiable by direct computation, \textit{coherence} is a much simple concept and preferable to use to assess the quality of a measure matrix [17]. In general, the smaller the coherence, the better the performance of spare recovery algorithms perform. 

\begin{definition} Let $\mathbf{A} \in \Bbb C^{m \times N}$ be a matrix with $l_2$ normalized columns
$\mathbf{a}_1,\cdots,\mathbf{a}_N$, i.e., $\|\mathbf{a}_i\|_2 = 1$ for all $i \in [N]$. The \textit{coherence} $\mu=\mu(A)$ of the matrix
$\mathbf{A}$ is defined as
$$\mu := \max\limits_{1\leq i\neq j\leq N} |\langle a_i, a_j \rangle|$$
\end{definition}

A general concept of \textit{$l_1$ - coherence function} is defined, which
incorporates the usual coherence as the particular value $s = 1$ of its argument.

\begin{definition}
Let $\mathbf{A} \in \Bbb C^{m \times N}$ be a matrix with $l_2$ - normalized columns $\mathbf{a}_1,\cdots, \mathbf{a}_N$. The \textit{$l_1$ - coherence} \textit{function $\mu_1$} of the matrix $\mathbf{A}$ is defined for $s \in [N-1]$
by
$$\mu_1(s) := \max\limits_{i \in [N]} \max \Big\{\sum\limits_{j \in S} |\langle \mathbf{a}_i, \mathbf{a}_j \rangle|, S \subset [N], card(S) = s, i \notin S\Big\}$$
\end{definition}

\begin{remark} Note that coherence and the $l_1$ - coherence function are invariant under multiplication on the left by unitary matrix $\mathbf{U}$ i.e columns of $\mathbf{UA}$ are $l_2$ -  normalized vectors $\mathbf{U}\mathbf{a}_1, \cdots, \mathbf{U}\mathbf{a}_N$ and $\langle \mathbf{U}\mathbf{a}_i, \mathbf{U}\mathbf{a}_j \rangle = \langle \mathbf{a}_i, \mathbf{a}_j \rangle$. Also using Cauchy-Schwarz inequality $|\langle \mathbf{a}_i, \mathbf{a}_j \rangle | \leq \|\mathbf{a}_i\|_2 \cdot \|\mathbf{a}_j\|_2$ it is clear that coherence matrix is bounded above $\mu \leq 1$
\end{remark}
What about lower bounds for the coherence and  $l_1$ - coherence
function of a matrix $\mathbf{A} \in \Bbb C^{m\times N}$ with $m < N$? What are examples of matrices with an almost minimal coherence? A matrix which achieve the coherence lower bound is called \textit{equiangular tight frame} and the coherence lower bound is known as the \textit{Welch bound}.

\begin{definition}
A system of $l_2$ - normalized vectors $(\mathbf{a}_1, \cdots, \mathbf{a}_N)$ in $\Bbb K^m$ is called \textit{equiangular} if there is a constant $c \geq 0$ such that
$$|\langle \mathbf{a}_i, \mathbf{a}_j \rangle| = c \text{ for all }i, j \in [N], i \neq j$$
\end{definition}

\begin{definition}
A system of vectors $(\mathbf{a}_1,\cdots, \mathbf{a}_N)$ in $\Bbb K^m$ is called a \textit{tight frame} if there exists a constant $\lambda > 0$ such that one of the following equivalent conditions holds:
\begin{itemize}
\item[(a)] $\|\mathbf{x}\|^2_2 = \lambda \sum\limits_{j=1}^N |\langle \mathbf{x}, \mathbf{a}_j\rangle |^2 \text{ for all } \mathbf{x} \in \Bbb K^m$
\item[(b)] $\mathbf{x}=\lambda \sum\limits_{j=1}^N \langle \mathbf{x}, \mathbf{a}_j \rangle \mathbf{a}_j \text{ for all } \mathbf{x} \in \Bbb K^m$
\item[(c)] $\mathbf{AA}^* = \frac{1}{\lambda} \mathbf{Id}_m \text{ where } \mathbf{A} \text{ is the matrix with columns } \mathbf{a}_1,\cdots, \mathbf{a}_N$.
\end{itemize}
\end{definition}

It is now possible to prove [7] that coherence of a matrix is always in the range $\mu(\mathbf{A}) \in \Big[\sqrt{\frac{N-m}{m(N - 1)}},1\Big]$

\begin{theorem} The coherence of a matrix $\mathbf{A} \in \Bbb K^{m\times N}$ with $l_2$ - normalized columns satisfies
\begin{equation}
\mu \geq \sqrt{\frac{N - m}{m(N-1)}}
\end{equation}
Equality holds if and only if the columns $\mathbf{a}_1, \cdots, \mathbf{a}_N$ of the matrix $\mathbf{A}$ form an equiangular tight frame.
\end{theorem}

The Welch bound can be extended to the $l_1$ - coherence function for small values
of its argument.

\begin{theorem} The $l_1$ - coherence function of a matrix $\mathbf{A} \in \Bbb K^{m\times N}$ with $l_2$ normalized columns satisfies
\begin{equation} \mu_1(s) \geq s\sqrt{\frac{N-m}{m(N-1)}} \text{ whenever }s <\sqrt{N - 1}
\end{equation}
Equality holds if and only if the columns $\mathbf{a}_1,\cdots, \mathbf{a}_N$ of the matrix $\mathbf{A}$ form an equiangular tight frame.
\end{theorem}

\begin{remark} Note that when $N >> m$ the lower bound is aproximately $\mu(\mathbf{A})\geq 1/\sqrt{m}$
\end{remark}

We claimed that the performance of sparse recovery algorithms is enhanced by a small coherence. Theorem below [7] guarantee the exact recovery of every $s$-sparse vector via basis pursuit when the measurement matrix has a coherence $\mu < 1/(2s - 1)$
\begin{theorem} Let $\mathbf{A} \in \Bbb C^{m \times N}$ be a matrix with $l_2$-normalized columns. If $$\mu_1(s) +\mu_1(s - 1) <1$$ then every $s$-sparse vector $\mathbf{x} \in \Bbb C^{N}$ is exactly recovered from the measurement vector $\mathbf{y = Ax}$ via basis pursuit
\end{theorem} 

Even though null space property is both necessary and sufficient condition to guarantee recovery of sparse vectors, when the measurements are contaminated with noise it will be useful to consider somewhat stronger conditions. Also lower bound on the coherence in Theorem 2.3.9. limits recovery algorithms to rather small sparsity levels. To overcome these limitations, in [18], Candes and Tao introduced the Restricted Isometry Property (RIP) on matrices $\mathbf{A}$  also known as uniform uncertainty principle and established its important role in compressed sensing. 

\begin{definition}The $s$th restricted isometry constant $\delta_s=\delta_s(\mathbf{A})$ of a matrix $\mathbf{A} \in \Bbb C^{m \times N}$ is a $\delta \geq 0$ such that 
\begin{equation}
(1-\delta)\|\mathbf{x}\|_2^2 \leq \|\mathbf{Ax}\|_2^2 \leq (1+\delta)\|\mathbf{x}\|_2^2
\end{equation}
for all $s$-sparse vectors $\mathbf{x} \in \Bbb C^N$. We say that $\mathbf{A}$ satisfies the \textit{restricted isometry property} if $\delta_s$ is small for reasonably large $s$.
\end{definition} Just like for coherence, small restricted isometry constants are desired. In case $N \geq Cm$ it have been proved in [7] that the restricted isometry constant must satisfy $\delta_s \geq c\sqrt{s/m}$ which is reminiscent of the Welch bound $\mu \geq c^{'}/\sqrt{m}$ when $s=2$.

The success of sparse recovery via basis pursuit for measurement matrices with small restricted isometry constants is suffice by the condition $\delta_{2s}<1/3$. Weakening this condition to $\delta_{2s} < 0.6246$ is actually sufficient to guarantee stable and robust recovery of all $s$-sparse vectors via $l_1$-minimization.

\begin{theorem} Suppose that the $2s$th restricted isometry constant of the matrix $\mathbf{A} \in \Bbb C^{m \times N}$ satisfies $$\delta_{2s} < \frac{1}{3}$$ Then every $s$-sparse vector $\mathbf{x} \in \Bbb C^{m \times N}$ is the unique solution of $$minimize_{\mathbf{z} \in \Bbb C^N} \|\mathbf{z}\|_1 \text{ subject to } \mathbf{Az=Ax}$$
\end{theorem}

\begin{theorem} Suppose that the $2s$th restricted isometry constant of the matrix $\mathbf{A} \in \Bbb C^{m \times N}$ satisfies $$\delta_{2s} < \frac{4}{\sqrt{41}} \sim 0.6246 $$ 
Then, for any $\mathbf{x} \in \Bbb C^N$ and $\mathbf{y} \in \Bbb C^m$ with $\|\mathbf{Ax - y}\|_2 \leq \eta$ a solution $\mathrm{x}^{\#}$ of 
$$minimize_{\mathbf{x} \in \Bbb C^N} \|\mathbf{z}\| \text{ subject to } \|\mathbf{Az-y}\|_2 \leq \eta$$
approximates the vector $\mathbf{x}$ with errors
$$\|\mathbf{x} - \mathbf{x}^{\#}\|_1 \leq C \sigma_s(\mathbf{x})_1 + D \sqrt{s}\eta$$
$$\|\mathbf{x} - \mathbf{x}^{\#}\|_2 \leq \frac{C}{\sqrt{s}} \sigma_s(\mathbf{x})_1 +D\eta$$ where the constants $C,D>0$ depend only on $\sigma_{2s}$.
\end{theorem}

\subsection{Sensing Matrix Constructions}
Now that we have defined the relevant properties of a matrix A in the context of CS, we turn to the question of how to construct matrices that satisfy these properties. We have already seen that in general equiangular tight frame achieves the coherence lower bound [22]. Similarly, there are known matrices of size $m \times m^2$ that achieve Welch bound (the coherence lower bound) $\mu(\mathbf{A})=1/\sqrt{m}$, such as the Gabor frame generated from the Alltop sequence [23]. It was proved that is possible to deterministically construct matrices of size $m \times N$ that satisfy the RIP of order $s$, but in real world settings these constructions would lead to an unacceptably large requirement on $m$. [24-25]. Fortunately, these limitations can be overcome by randomizing the matrix construction. Random matrices will satisfy the RIP with high probability if the entries are chosen according to a Gaussian, Bernoulli, or more generally any sub-gaussian distribution [7]. Theorem 5.65 in [8] states that if a matrix $\mathbf{A}$ is chosen according
to a sub-gaussian distribution with $m =O(s\log{(N/s)}/\delta_{2s}^2)$ then $\mathbf{A}$ will satisfy the RIP of order $2s$ with probability at least $1 - 2 exp(-c_1\delta{2s}^2 m)$. The most significant benefits of using random matrices is met in practice where we are often more interested in recovering spare signal with respect to some basis $\Phi$ thus we require the product $A\Phi$ to satisfy the RIP. When $\mathbf{A}$ is chosen randomly we do not have to explicitly take $\Phi$ into account.

\section{Restricted isometric projections for Riemannian manifolds} 

Given a set of observation vectors $\Sigma$ embedded in the Euclidean space $\mathbb{R}^N$, where 
$\# \Sigma < N$, we wish to be able to compare various instances of the restricted isometric projection of $\Sigma$ on $m-$dimensional linear subspaces of $\mathbb{R}^N$, $m \ll N$, and establish if the set $\Sigma$ may be associated to a Riemannian manifold $(\mathcal{M}, g)$, of manifold dimension $m$, with Riemannian metric $g$ equivalent to the induced metric from the embedding space $\mathbb{R}^N$. The purpose behind formulating this question is that of establishing a higher-dimensional version of the Fisher-Kolmogorov test for comparing populations in usual statistical analysis, or (alternatively) to develop an inference procedure analogous to GLM (Generalized Linear Models) in the usual case ($\# \Sigma \gg N$). If successful, the association $\Sigma \to (\mathcal{M}, g)$ would allow to establish an obvious equivalence relation between two distinct sets of vectors $\Sigma_1, \Sigma_2$, once they are associated to the same manifold. 

In the following section we formulate the fundamental problem and present a classification criterion. 

\subsection{Generalized Restricted Isometric Projections}

In the following, we take positive integers $m, n, N$ to be related by $n < N, \, m \ll N$. 

\bigskip

\noindent {\bf{Fundamental problem.}} {\emph{
Let $\Sigma = \{v_1, v_2, \ldots, v_n \}$ be embedded in the Euclidean space $\mathbb{R}^N$,  such that there exists a restricted isometric projection to a hyperplane $\mathcal{H} \simeq \mathbb{R}^m$, with distortion factor $0 < \delta \ll 1$.  Is there a Riemannian manifold $(\mathcal{M}, g)$, with $\dim \mathcal{M} = m$, and a point $P \in \mathcal{M}$, such that 
\begin{equation}
\mathcal{H} = T_P  \mathcal{M}, \,\, || \, . \,\,  ||_{\ell^2 (\mathbb{R}^N)} \xhookrightarrow{}  g( \, . \,), \,\, 
\varphi_P (\Sigma) = \widehat{\Sigma} \subset \mathcal{M}, 
\end{equation}
and the matrix of pairwise distances between the elements of $\widehat{\Sigma}$, in the metric $g$, has distortion $O(\delta)$, where $\varphi_P$ is the inverse local coordinate chart $\varphi_P:  T_P \, \mathcal{M} \to \mathcal{M}$? 
}}

\begin{remark}
An obvious extension of the problem would only require identifying the manifold $\mathcal{M}$ up to an isometry. 
\end{remark}

\begin{remark}
If a set $\Sigma \subset \mathbb{R}^N$ can be associated with a Riemannian manifold $(\mathcal{M}, g)$ as described in the Fundamental Problem, then we say that has the {\emph{extended restricted isometric property}}, and $\widehat{\Sigma}$ is an extension of $\Sigma$. 
\end{remark}

\begin{theorem}Assume that an arbitrary set of vectors $\Sigma \subset \mathbb{R}^N$ 
has the {\emph{extended restricted isometric property}} with projection to $\mathbb{R}^m$. Then the Riemannian manifold $(\mathcal{M}, g)$ is (up to a global dilation) the symmetric space $SO(m+1)/(SO(1)\times SO(m))$. 
\end{theorem}

\begin{proof}
Assume there exists a differentiable function $F: D \to \mathbb{R}^N$, where $D \subset \mathbb{R}^m$ is  an open, simply-connected set of full measure, such that $F$ is a diffeomorphism between $D$ and $F(D)$. Then the induced metric on the cotangent space at $p \in F(D)$ is given in parametric form as 
$$
g \in T_p^* F(D) \odot T_p^* F(D), \,\, g(V,W) = \sum_{1 \le i, j \le m} g_{ij} V^i W^j, 
$$ 
where $V, W \in T_p F(D)$ are vectors from the tangent space, and 
$$
g_{ij} = \Big \langle \frac{\p F}{\p t^i}, \frac{\p F}{\p t^j} \Big \rangle,
$$
with $\{t^i \}_{i = 1}^m$ coordinates in $D$, and $\langle \, , \,\rangle$ the usual scalar product on $\mathbb{R}^N$. 

If the collection of vectors $\Sigma  = \{ v_1, \ldots, v_n \} \in D$ are the result of projecting the original set of vectors $\Sigma_0 = \{ x_1, \ldots, x_n \}$ from $\mathbb{R}^N$ to $\mathbb{R}^m$, with restricted isometry constant $\delta$, by application of compressions $A_k, \, k= 1, 2, \ldots, n$, where $A_k \in \mathbb{R}^{m\times N}$, then the metric tensor evaluated at the point $p_k = F(v_k) \in F(D), \, k = 1, 2, \ldots, n$ takes the form 
$$
g_{ij}(p_k) =  \Big \langle \frac{\p F}{\p t^i}(v_k), \frac{\p F}{\p t^j}(v_k) \Big \rangle, 
$$ 
or in matrix form 
$$
G_k = (DF)^T(v_k)\cdot DF(v_k), 
$$
where 
$$
DF \in \mathbb{R}^{N \times m}
$$
is the derivative matrix of $F$, $(DF)^T$ is its transpose, and $G_k \in \mathbb{R}^{m \times m}$ is the metric matrix at $p_k$. If the diffeomorphism $F$ coincides with the inverse transformation $v_k \to x_k$ when evaluated on $\Sigma$, then we obtain the set of matrix conditions
$$
A_k \cdot (DF)_k = \mathbb{I}_{m \times m}, \, \forall k = 1, 2, \ldots, n. 
$$
This shows that, if the compression matrices $\{ A_k \}_{k=1}^n$ are i.i.d. from the same ensemble of random matrices, then the matrices $(DF)_k$ are also i.i.d. with the distribution given by the generalized inverse of $A_k$, and therefore the matrices $\{ G_k \}_{k = 1}^m$ are also i.i.d. covariance matrices, obviously positive-definite, and invertible with probability 1. Therefore, 
$$
G_k \stackrel{i.i.d.}{\sim} G, \,\, \forall k = 1, 2, \ldots, n,
$$
where $G$ is a diagonal, positive-definite matrix, and all $G_k$ are obviously in its conjugacy class within $GL(n, \mathbb{R})$. 

Denote by $\widehat{G}$ the isotropy group of $G$, then since the set $\Sigma$ was chosen arbitrary, we conclude that the Riemannian manifold $F(D)$ has the metric isotropy group $\widehat{G}$ acting transitively, and therefore $F(D)$ must be on open subset of a symmetric space $U/\widehat{G}$. Therefore, we can use the Cartan classification of Riemannian symmetric spaces to distinguish two possible cases: either $F(D)$ has zero curvature, and is therefore an Euclidean space, or it has positive curvature, and the manifold is then of compact type, i.e. equivalent to the quotient of two real Lie groups, $U$ and $\widehat{G}$. 

It remains to identify the possible choices of real Lie groups $U, \widehat{G}$ in the Cartan classification of compact symmetric spaces, compatible both with the requirement that $\widehat{G}$ belong to an invariance group for random covariance matrices, and the dimensional constraint $\dim U/\widehat{G} = m$. 

Together with the condition that $U, \widehat{G}$ be real Lie groups, the dimension constraint implies that 
$$
U/\widehat{G} \simeq SO(m+1)/(SO(1)\times SO(m)), \,\, \dim U/ \widehat{G} = m,
$$
and that the ensemble of covariance matrices $\{ A^T_k \cdot A_k \}$ is invariant under the induced action of $SO(m)$. 
\end{proof}

\begin{remark}
Assume that all entries in compression matrices $A_k$ are i.i.d. Gaussian. Then the eigenvalue distribution of the covariance matrices has the large $m$-limit ($m \gg 1$) given by the Marchenko-Pastur law. Then for $m \to \infty$, $\Sigma$ does have the extended restricted isometric property. 
\end{remark}
 
 \bibliographystyle{amsplain}

\noindent [1] D. Donoho. \textit{Compressed sensing.} IEEE Trans. Inform. Theory, 52(4):1289{
1306, 2006.

\noindent [2] B. Adcock, A. C. Hansen, C. Poon, and B. Roman, \textit{Breaking the coherence
barrier: a new theory for compressed sensing}, arXiv:1302.0561, 2013.

\noindent
[3] E. J. Candes, \textit{Compressive sampling}, in Proc. International Congress
of Mathematicians, Madrid, Spain, 2006, vol. 3.

\noindent
[4] E. Candes, J. Romberg, and T. Tao, \textit{Robust uncertainty principles: Exact signal reconstruction from highly incomplete frequency information.}, IEEE Transactions on information theory, vol. 52, pp. 489–509, Feb 2006.

\noindent
[5] E. Candes and J. Romberg, \textit{Quantitative robust uncertainty principles and optimally sparse decompositions}, Foundations of Comput. Math, vol. 6, no. 2, pp. 227 – 254, 2006.

\noindent
[6] E. Candes and T. Tao, \textit{Near optimal signal recovery from random projections: Universal encoding strategies?,} IEEE Trans. on Information Theory, vol. 52, no. 12, pp. 5406 – 5425, 2006.

\noindent
[7] Foucart, S., and Rauhut, H. \textit{A mathematical introduction to Compressed Sensing},
Birkhauser, 2013.

\noindent
[8] M. A. Davenport, M. F. Duarte, Y. C. Eldar, and G. Kutyniok. Introduction to compressed sensing. In Compressed Sensing: Theory and Applications. Cambridge University Press, 2011.

\noindent
[9] T. Strohmer. \textit{Measure what should be measured: progress and challenges in compressive sensing}. IEEE Signal Process, 2012.

\noindent
[10] M. Duarte and Y. Eldar, \textit{Structured compressed sensing: From theory
to applications,} IEEE Trans. Sig. Proc., vol. 59, no. 9, pp. 4053–4085, 2011

\noindent
[11] G. Tang, B. Bhaskar, P. Shah, and B. Recht. \textit{Compressed sensing off the grid.} Preprint, 2012.

\noindent
[12] E. J. Candes, J. Romberg, and T. Tao. \textit{Robust uncertainty principles: exact signal reconstruction from highly incomplete frequency information}. IEEE Trans. Inform. Theory, 52(2):489–509, 2006.

\noindent
[13] M. Lustig, D. L. Donoho, J. M. Santos, and J. M. Pauly. \textit{Compressed Sensing MRI}. Signal Processing Magazine, IEEE, 25(2):72–82, March 2008.

\noindent
[14] B. Adcock, A. C. Hansen, B. Roman, and G. Teschke. \textit{Generalized sampling: stable reconstructions, inverse problems and compressed sensing over the continuum}. Advances in Imaging and Electron Physics, 182:187–279,2014.

\noindent
[15] R. Ward. \textit{Compressive sensing with cross validation}. IEEE Trans Inform Theory,
55(12):5773–5782, 2009.

\noindent
[16] R. Tibshirani. \textit{Regression shrinkage and selection via the Lasso}. J Roy Statist Soc B, 58(1):267–288, 1996.

\noindent
[17] D. Donoho and M. Elad. \textit{Optimally sparse representation in general (nonorthogonal) dictionaries via $l_1$ minimization.} Proc Natl Acad Sci, 100(5):2197–2202, 2003.

\noindent
[18] E. Candes and T. Tao. \textit{Decoding by linear programming}. IEEE Trans Inform Theory, 51(12):4203–4215, 2005.

\noindent
[19] Marco F. Duarte, Shriram Sarvotham, \textit{Joint Sparsity Models for Distributed Compressed Sensing}

\noindent
[20] Yonina C. Eldar,Patrick Kuppinger, \textit{Block-Sparse Signals: Uncertainty Relations and Efficient Recovery, IEEE TRANSACTIONS ON SIGNAL PROCESSING, VOL. 58, NO. 6, JUNE 2010
}
\noindent
[21] R. G. Baraniuk, V. Cevher, M. F. Duarte, and C. Hegde. \textit{Model-based compressive sensing}

\noindent
[22] T. Strohmer and R. Heath. \textit{Grassmanian frames with applications to coding and communication.} Appl Comput Harmon Anal, 14(3):257–275, 2003.

\noindent
[23] M. Herman and T. Strohmer. \textit{High-resolution radar via compressed sensing.} IEEE Trans Sig Proc, 57(6):2275–2284, 2009.

\noindent
[24] R. DeVore. \textit{Deterministic constructions of compressed sensing matrices.} J Complex, 23(4):918–925, 2007.

\end{document}